\documentclass{amsart}
\usepackage{amsmath}
  \usepackage{paralist}
  \usepackage{graphics} 
  \usepackage{epsfig} 
 \usepackage[colorlinks=true]{hyperref}
\hypersetup{urlcolor=blue, citecolor=red}

\newtheorem{theorem}{Theorem}[section]

\newtheorem{lemma}[theorem]{Lemma}

\theoremstyle{definition}

\newcommand{\Z}{{\mathbb Z}}
\newcommand{\R}{{\mathbb R}}
\newcommand{\C}{{\mathbb C}}
\def\calI{\mathcal I}
\def\calL{\mathcal L}

\def\calM{\mathcal M}
\def\O{\mathcal O}

\def\wt{\widetilde}

\def\ol{\overline}
\def\i{{\mathrm i}}
\def\om{{\omega}}
\def\th{{\theta}}
\def\vth{{\vartheta}}
\def\E{\mathbf{E}}
\def\boe{\mathbf{E}}

\def\ve{\varepsilon}
\def\tu{{\tt u}}
\def\cc{\mathrm{c.c.}}

\def\res{{\mathrm{res}}}

\title[Modulated waves in diatomic chains]
{Transport and generation of  macroscopically modulated waves 
in diatomic chains}

\author[Johannes Giannoulis]{}


\email{giannoul@ma.tum.de}

\begin{document}
\maketitle

\centerline{\scshape Johannes Giannoulis}
\medskip
{\footnotesize
 \centerline{Department of Mathematics, TU M\"unchen}
 \centerline{Boltz\-mann\-str.~3, D-85747 Garching, Germany}
}


\bigskip


\begin{abstract}
We derive and justify analytically the dynamics of a small macroscopically modulated amplitude of a single plane wave in a nonlinear diatomic chain with stabilizing on-site potentials
including the case where a wave generates another wave via self-interaction. 
More precisely, we show that in typical chains acoustical waves can generate 
optical but not acoustical waves, while optical waves are always closed with respect to self-interaction. 
\end{abstract}

\section{Introduction}
The present work constitutes a generalization of previous work of the author, see \cite{Gia10},
to a case of vector-valued displacement in nonlinear lattices. As the technically most simple but yet 
generic case we consider a nonlinear diatomic chain. For the physical derivation, interpretation and discussion of several applications of the harmonic diatomic chain
we refer to \cite{Bri53}. 
Various questions concerning diatomic lattices have been addressed up to now, see e.g.\ 
\cite{BG93, GEM97,PFR86,Tsu73,YS79}.
Here we focus on the analytical justification of the dynamics of small macroscopic amplitude modulations, see \eqref{appr}. 
More precisely,
we consider the diatomic chain
\begin{align}\label{dichain}
\begin{cases}
\ddot x_{2j+1}
&=V_1^\prime(x_{2j+2}-x_{2j+1})-V_1^\prime(x_{2j+1}-x_{2j})
-W_1^\prime(x_{2j+1}),
\\
\ddot x_{2j}
&=V_2^\prime(x_{2j+1}-x_{2j})-V_2^\prime(x_{2j}-x_{2j-1})
-W_2^\prime(x_{2j}),
\end{cases}
\quad j\in\Z,
\end{align}
with nearest-neighbor interaction and on-site potentials $V_i,\,W_i\in C^4(\R)$, $i=1,2$, 
such that
\begin{equation}\label{TEVW}
\begin{cases}
V_i^\prime(x) &=v_{i,1}x+v_{i,2}x^2+\widetilde V_i^\prime(x),
\quad \widetilde V_i^\prime (x)=O(|x|^3),
\\
W_i^\prime(x) &=w_{i,1}x+w_{i,2}x^2+\widetilde W_i^\prime(x),
\quad \widetilde W_i^\prime (x)=O(|x|^3).
\end{cases}
\end{equation}
Setting $u_j=\begin{pmatrix} u_{j,1} \\ u_{j,2} \end{pmatrix} 
:=\begin{pmatrix}x_{2j+1}\\x_{2j}\end{pmatrix}$, $j\in\Z$,  
and using the Taylor-expansions \eqref{TEVW},
the diatomic chain \eqref{dichain} takes the form 
\begin{align}\label{dichainS}
\ddot u = &\  \calL u+\calM(u),
\displaybreak[0]\\
(\calL u)_j :=&\ 
\begin{pmatrix}
v_{1,1} \big(u_{j+1,2}-2 u_{j,1}+u_{j,2}\big) -w_{1,1} u_{j,1}
\notag\\[1mm]
v_{2,1} \big(u_{j,1}-2 u_{j,2}-u_{j-1,1}\big) -w_{2,1} u_{j,2}
\end{pmatrix},
\displaybreak[0]\\
\notag
(\calM(u))_j :=&\ 
\begin{pmatrix}
v_{1,2} \big((u_{j+1,2}-u_{j,1})^2-(u_{j,1}-u_{j,2})^2\big)-w_{1,2}u_{j,1}^2
\\[1mm]
v_{2,2} \big((u_{j,1}-u_{j,2})^2-(u_{j,2}-u_{j-1,1})^2\big)-w_{2,2}u_{j,2}^2
\end{pmatrix}
+
\displaybreak[0]\\& 
+\begin{pmatrix}
\widetilde V_1^\prime (u_{j+1,2}-u_{j,1})-\widetilde V_1^\prime (u_{j,1}-u_{j,2})
-\widetilde W_1^\prime (u_{j,1})
\\[1mm]
\widetilde V_2^\prime(u_{j,1}-u_{j,2})-\widetilde V_2^\prime (u_{j,2}-u_{j-1,1})
-\widetilde W_2^\prime (u_{j,2})
\end{pmatrix}.
\notag
\end{align}
The linearized model $\ddot u = \calL u$ 
admits for non-trivial plane-wave solutions 
\begin{align*}
u= A  \E +\cc
,\quad \boe(t,j):=e^{\i(\om t +j\vth)}
,\quad A:=\begin{pmatrix} A^{(1)}\\ A^{(2)} \end{pmatrix}\in \C^2,
\end{align*}
provided the frequency $\om\in\R$ and the wave number $\vth\in(-\pi,\pi]$ satisfy the dispersion relation 
\begin{align}\label{H}
\det H(\om,\vth)=0
,\quad H(\om,\vth):=
\begin{pmatrix} 
\om^2{-}c_1 &  v_{1,1}(e^{\i\vth}{+}1) 
\\ 
v_{2,1}(1{+}e^{-\i\vth}) & \om^2{-}c_2
\end{pmatrix}
,
\end{align}
where $c_i:=2 v_{i,1}+w_{i,1}$.
This is equivalent to 
\begin{equation}\label{disprelvth}
\om^2=
\om_\pm^2(\vth):=\frac{c_1+c_2}2\pm\frac12\sqrt{(c_1-c_2)^2+ 8 v_{1,1} v_{2,1} (\cos\vth+1)}.
\end{equation}
Assuming $c_1+c_2>0$, $c_1c_2>4v_{1,1}v_{2,1}>0$,  
we obtain 
\begin{align}\label{freq}
\om_\pm(\vth)
:=
\sqrt{\frac12\left(c_1+c_2\pm\sqrt{(c_1-c_2)^2+8 v_{1,1} v_{2,1} (\cos\vth+1 )}\right)}
>0
\end{align}
for all $\vth\in (-\pi,\pi]$ and the additional assumption $c_1\neq c_2$ yields the strict separation of the optical and acoustical branches of the frequency, 
\begin{align*}
2\om_+^2(\vth)\ge c_1+c_2+|c_1-c_2|>c_1+c_2-|c_1-c_2|\ge 2\om_-^2(\vth)
\quad \forall\ \vth\in(-\pi,\pi].
\end{align*}
All of the above assumptions are satisfied in the case 
$w_{i,1}>0$, $4 v_{i,1}+w_{i,1}>0$, $v_{1,1}v_{2,1}>0$,  
$ 2 v_{2,1}+w_{2,1}>2 v_{1,1}+w_{1,1}$,
which we assume in the following.

The 
eigenvectors $A$ to the eigenfrequencies 
$\om=\om_\pm(\th)$ are given by
\begin{equation}\label{relations}
A^{(2)}=-\rho A^{(1)},
\qquad 
\rho:=\frac{\om^2{-}c_1}{v_{1,1}(e^{\i\vth}{+}1)}
=\frac{v_{2,1}(e^{-\i\vth}{+}1)}{\om^2{-}c_2}\neq0,
\quad\text{if $\vth\neq\pm\pi$}
\end{equation}
and 
\begin{equation}\label{vthpi}
A=\begin{pmatrix} A^{(1)} \\ 0 \end{pmatrix}
\quad\text{for $\om=\om_-(\pm\pi)$,}
\quad 
A=\begin{pmatrix} 0 \\ A^{(2)} \end{pmatrix}
\quad\text{for $\om=\om_+(\pm\pi)$.}
\end{equation}
The plan of the paper is as follows. In Section \ref{resonances} we discuss whether a given plane wave solution $\boe$ can generate via self-interaction another plane wave $\boe^2$. Then, taking into account also this possibility, in Section \ref{formalderivation} we derive formally the macroscopic equations for the first order amplitudes $A_{1,n}$ of two waves $n=1,2$, and finally, in Section \ref{justification}, we justify the derived equations.

\section{Resonances}\label{resonances}
Since we are interested in the self-interaction of a plane wave $\boe$, which means that  $\boe^2$ is also a plane wave, in a diatomic chain we are interested in resonance conditions like the ones on the left hand side below.  
Making in \eqref{freq} the substitutions 
$c:=(\cos\vth+1)/2\in[0,1]$, $d_1:=(c_1+c_2)^2/f>0$, $d_2:=(c_1-c_2)^2/f>0$ 
with $f:=16v_{1,1}v_{2,1}>0$ and $d_1-d_2>1$,
the problem of finding a $\vth\in(-\pi,\pi]$ satisfying one of these resonance
conditions is equivalent to
finding a $c\in[0,1]$ for given 
$d_1>d_2+1>1$ 
satisfying the corresponding equation on the
right hand side: 
\begin{align}
2\om_{(\pm)}(\vth)=\om_\pm(2\vth)
&\quad\Leftrightarrow\quad
4\left(\sqrt{d_1}\,(\pm)\,\sqrt{d_2+c}\right)=\sqrt{d_1}\pm\sqrt{d_2+(2c-1)^2}
\notag\\
&\quad\Leftrightarrow\quad
3\sqrt{d_1}=\,(\mp)\,4\sqrt{d_2+c}\pm\sqrt{d_2+(2c-1)^2}.
\label{rescond}
\end{align}
By the positivity of all appearing square roots we immediately see that a
resonance $2\om_+(\vth)=\om_-(2\vth)$, i.e., an optical wave generating an
acoustical one, is not possible. Moreover, since
\begin{align*}
3\sqrt{d_1}>\sqrt{d_1}>
\sqrt{d_2+1} 
>
-4\sqrt{d_2+c}+\sqrt{d_2+(2c-1)^2},
\end{align*}
we see 
that an optical wave can not generate another optical one, 
i.e., $2\om_+(\vth)\neq \om_+(2\vth)$ $\forall$ $\vth\in(-\pi,\pi]$.
Thus, 
{\it an optical wave is closed under self-interaction of order $2$.}

However, {\it an acoustical wave can generate an optical one by 
self-interaction}, i.e., for appropriate choice of the harmonic parts of the
interaction and on-site potentials there exist $\vth\in(-\pi,\pi]$ such that
$2\om_-(\vth)=\om_+(2\vth)$. After taking squares on the left and right hand 
sides, the corresponding condition \eqref{rescond} reads
\begin{align}\label{rescondopt}
9d_1=17d_2+16c+(2c-1)^2+8\sqrt{d_2+c}\sqrt{d_2+(2c-1)^2},
\end{align}
and we want to prove the existence of a $c\in[0,1]$ that 
satisfies this condition for the 
$d_1,d_2$ given above. 
We restrict ourselves 
to the case 
$v_{1,1}=a>0$, $v_{2,1}=\gamma a$, $\gamma>1$, $w_{1,1}=w_{2,1}=b>0$. 
This setting 
satisfies all conditions posed so far on the harmonic coefficients,
and we 
obtain
 \begin{align}\label{dvalues}
 & 
 d_1 
=\frac{(\gamma+1)^2}{4\gamma}
+\frac1\gamma\left((\gamma+1)\frac{b}{a}+\frac{b^2}{a^2}\right),
 \qquad
 d_2=\frac{(\gamma-1)^2}{4\gamma}=:\delta
 \end{align}
(which obviously satisfies $d_1>d_2+1>1$). 
Inserting these values into \eqref{rescondopt}, we get
\begin{align*}
\frac9\gamma\left((\gamma+1)\frac{b}{a}+\frac{b^2}{a^2}\right)
=
8\delta-9+16c+(2c-1)^2+8\sqrt{\delta+c}\sqrt{\delta+(2c-1)^2}.
\end{align*}
Hence, for every $c\in[0,1]$ such that 
$16c\ge 9-8\delta$ 
there exists a $\frac{b}a$ such that 
\eqref{rescondopt} is satisfied.
Since $\delta>0$, we can always find such a $c$. 

Furthermore, the resonance condition for the generation of an acoustical
wave from an acoustical one, $2\om_-(\vth)=\om_-(2\vth)$, is equivalent to 
\begin{align*}
3\sqrt{d_1}&=4\sqrt{d_2+c}-\sqrt{d_2+(2c-1)^2}
\end{align*}
Concerning the case just considered, we observe that for $d_2=\delta$, 
the r.h.s.\ is nonnegative only for $c\in[c_e,1]$ with 
$c_e:=\max\{0,\frac{5-\sqrt{15\delta+24}}2\}$ 
(and hence for all $c\in[0,1]$ when $\delta\ge 1/15$).
Restricting our analysis to the set $[c_e,1]$ (non-empty for all $\delta>0$),
we obtain by squaring and insertion of the values \eqref{dvalues} as above 
\begin{align*}
\frac9\gamma\left((\gamma+1)\frac{b}{a}+\frac{b^2}{a^2}\right)
=8\delta-9+16c+(2c-1)^2-8\sqrt{\delta+c}\sqrt{\delta+(2c-1)^2},
\end{align*}
although with a minus sign in front of the square root.
Due to the existent on-site potential (where $b>0$), 
in order to obtain resonances the r.h.s.\ $g$ needs to be strictly positive
for some $c\in[c_e,1]$. However, a careful analysis reveals that  
$g(c)\le 0$ for $c\in[c_e,1]$, 
and we obtain that in the case 
$v_{2,1}=\gamma a>a=v_{1,1}$, $w_{1,1}=w_{2,1}=b>0$, {\it an acoustical wave can not generate another acoustical one by self-interaction}. 

Finally, we conclude by showing that $\om_-(\vth)+\om_+(2\vth)\neq \om_+(3\vth)$ 
for all $\vth\in(-\pi,\pi]$. Indeed, after squaring the left and right hand sides we see that the equality is equivalent to  
\begin{align*}
&
-\sqrt{d_2+c}
+\sqrt{d_2+(2c-1)^2}
+2\sqrt{\big(\sqrt{d_1}-\sqrt{d_2+c}\big)\big(\sqrt{d_1}+\sqrt{d_2+(2c-1)^2}\big)}
\\&\quad
=\sqrt{d_2+(4 c-3)^2 c}-\sqrt{d_1}
\end{align*}
for $d_1>d_2+1>1$. Since $(4 c-3)^2 c=(\cos(3\vth)+1)/2\in[0,1]$, the r.h.s.\ of this equation is always $<0$, and 
it suffices to show that the l.h.s.\ is $\ge 0$ even for $c\ge (2c-1)^2$.
Hence, since $d_1>d_2+1$, it is sufficient to show that the l.h.s.\ with $\sqrt{d_1}$ replaced by $\sqrt{d_2+1}$ is $\ge 0$ for $c\in[1/4,1]$. 
Comparing in this modified l.h.s.\ the square of the first two terms with the square of the third one, and adding a suitable term, this is equivalent to showing that 
\begin{align*}
4 (1-c)\ge 
&\big(\sqrt{d_2+c}-\sqrt{d_2+(2 c-1)^2}\big)
\\
&\Big(\big(\sqrt{d_2+c}-\sqrt{d_2+(2 c-1)^2}\big)+4 \big(\sqrt{d_2+1}-\sqrt{d_2+c}\big)\Big)
\end{align*}
for $c\in[1/4,1]$. 
Since the r.h.s.\ is positive and strictly decreasing as a function of $d_2>0$ 
when $c>(2c-1)^2$, it suffices to show 
\begin{align*}
g(c):=4 (1-c)-
\big(\sqrt{c}-\sqrt{(2 c-1)^2}\big)
\Big(\big(\sqrt{c}-\sqrt{(2 c-1)^2}\big)+4 (1-\sqrt{c})\Big)
\ge0
\end{align*}
for $c\in[1/4,1]$, which holds true (as an elementary analysis shows),
with $g(1)=0$. 

\section{Formal derivation}\label{formalderivation}
We are interested in solutions of \eqref{dichainS} which in first order in $\ve$ are a sum of two macroscopically modulated plane-wave solutions with small amplitudes
\begin{equation}\label{appr}
u=U^{A,1}_\ve+O(\ve^2)
,\quad (U^{A,1}_\ve)_{j}(t):=\ve \sum_{n=1}^2 A_{1,n}(\ve t,\ve j)\boe_n(t,j)+\cc
\end{equation}
where $A_{1,n}=\big(A_{1,n}^{(1)}, A_{1,n}^{(2)}\big)^T :\R\times\R\to\C^2$ and 
$\E_n(t,j):=e^{\i(\om_n t+j\vth_n)}$ with $(\om_n,\vth_n)$ satisfying \eqref{H}.

However, due to the scaling of $A_{1,n}$ by $\ve$ and the macroscopic nature of its time and space variables, its dynamics will include terms of second order in $\ve$. 
Hence, taking into account the nonlinearity of our original system \eqref{dichainS} and the fact that we consider two different plane waves, we insert into \eqref{dichainS} the 
\emph{improved approximation}
\begin{align}
\label{imprappr}
U^{A,2}_\ve:= U^{A,1}_\ve 
+\ve^2 \Big(
&\sum_{n=1}^2 \big(A_{2,n}\boe_n+ A_{2,(n,n)}\boe_n^2\big)
+A_{2,(1,2)}\boe_1\boe_2
\\&
+A_{2,(1,-2)}\boe_1\boe_{-2}
+\frac12 A_{2,(1,-1)}+\cc\Big),
\notag
\end{align}
where $A_{2,\iota}=\big(A_{2,\iota}^{(1)},A_{2,\iota}^{(2)}\big)^T:\R\times\R\to\C^2$,
$\iota\in\{1,2\}\cup I$, $I:=\{(1,1),\, (2,2),\, (1,2),$ 
$(1,-2),\, (1,-1)\}$,
are again functions of the macroscopic variables $\tau=\ve t$, $y=\ve j$,
and where $\E_{-n}=\ol\E_n$.
Thereby, we use 
the Taylor expansions 
\begin{align*}
&A_{1,n}^{(i)}(\cdot,\cdot\pm\ve)
=A_{1,n}^{(i)}\pm\ve\partial_y A_{1,n}^{(i)}
+\ve^2\frac12\partial_y^2 A_{1,n,\xi_1+}^{(i)},\quad
\partial_y^2 A_{1,n,\xi_1\pm}^{(i)}:=\partial_y^2 A_{1,n}^{(i)}(\tau,y\pm\xi_1\ve),
\\
&A_{2,\iota}^{(i)}(\cdot,\cdot\pm\ve)
=A_{2,\iota}^{(i)}\pm\ve\partial_y A_{2,\iota,\xi_2\pm}^{(i)},
\quad
\partial_y A_{2,\iota,\xi_2\pm}^{(i)}:=\partial_y A_{2,\iota}^{(i)}(\tau,y\pm\xi_2\ve)
\end{align*}
with $\xi_1,\xi_2\in(0,1)$, assuming $A_{1,n}(\tau,\cdot)\in C^2(\R;\C^2)$, 
$A_{2,\iota}(\tau,\cdot)\in C^1(\R;\C^2)$.

Carrying out the usual (lengthy but straightforward) formal expansion 
in terms of $\ve$ and $\E_n$, we obtain that 
$ \ddot U^{A,2}_\ve= \calL U^{A,2}_\ve+\calM\big(U^{A,2}_\ve\big)$
is equivalent to
\begin{align*}
&\ve\Big\{
\sum_{n=1}^2 H(\om_n,\vth_n)
A_{1,n}\boe_n+\cc\Big\} +
\displaybreak[0]
\\&
+\ve^2\Big\{
\sum_{n=1}^2 \Big(
\begin{pmatrix}
-2\i\om_n \partial_\tau A^{(1)}_{1,n} +v_{1,1}e^{\i\vth_n}\partial_y A_{1,n}^{(2)}
\\
-2\i\om_n \partial_\tau A^{(2)}_{1,n} - v_{2,1}e^{-\i\vth_n} \partial_y A_{1,n}^{(1)}
\end{pmatrix}
+H(\om_n,\vth_n)
A_{2,n}
\Big)\boe_n
\displaybreak[0]\\&\phantom{+\ve^2\Big\{\  }
+\sum_{n=1}^2 \Big( 
H(2\om_n,2\vth_n)
A_{2,(n,n)}
+K_{(n,n)}
\Big)\boe_n^2
\displaybreak[0]\\&\phantom{+\ve^2\Big\{\  }
+\Big(H(\om_1{+}\om_2,\vth_1{+}\vth_2)
A_{2,(1,2)}
+K_{(1,2)}
\Big)\boe_1\boe_2
\displaybreak[0]\\&\phantom{+\ve^2\Big\{\  }
+\Big(
H(\om_1{-}\om_2,\vth_1{-}\vth_2)
A_{2,(1,-2)}
+K_{(1,-2)}
\Big)\boe_1\boe_{-2}
\displaybreak[0]\\&\phantom{+\ve^2\Big\{\  }
+\frac12
H(0,0)
A_{2,(1,-1)}
+K_{(1,-1)}
+\cc\Big\}
+\res\big(U^{A,2}_\ve\big)
=0
\end{align*}
with the explicit expressions for $K_{\iota}$, $\iota\in I$, 
and $\res\big(U^{A,2}_\ve\big)=\O(\ve^3)$ given in the Appendix.
Hence, in order for our ansatz \eqref{imprappr} to satisfy \eqref{dichainS} up to order $\ve$,
taking into account that $\boe_1\neq\boe_2$, the systems $H(\om_n,\vth_n)A_{1,n}=0$ have to be satisfied. As we have already seen, since $\det H(\om_n,\vth_n)=0$, this gives the relation between first and second component of $A_{1,n}$
\eqref{relations}, \eqref{vthpi} with $A,\rho,\om,\vth$ replaced by $A_{1,n},\rho_n,\om_n,\vth_n$.

Next, we assume that 
\begin{equation}\label{regularity}
\det H(\om,\vth)\neq 0\quad\text{for}\quad  
(\om,\vth)=(2\om_n,2\vth_n),\, (\om_1\pm\om_2,\vth_1\pm\vth_2),
\end{equation}
which means in particular that $\boe_n^2,\, \boe_1\boe_2,\, \boe_1\boe_{-2}\neq \boe_1,\boe_2$.
(Note here that $\det H(0,0)\neq0$ is always satisfied due to our stability assumption $c_1c_2>4v_{1,1}v_{2,1}$.)
In this case and for $\vth_n\neq\pm\pi$, 
we obtain from the equations for $\ve^2\boe_n$
\begin{align}\label{relations2}
\rho_n A_{2,n}^{(1)} +  A_{2,n}^{(2)}
&=\frac1{v_{1,1}(e^{\i\vth_n}{+}1)}
\big(2\i\om_n \partial_\tau A^{(1)}_{1,n}-v_{1,1}e^{\i\vth_n}\partial_y A_{1,n}^{(2)}\big)
\\&
=\frac1{\om_n^2{-}c_2}
\big( 2\i\om_n \partial_\tau A^{(2)}_{1,n} +v_{2,1}e^{-\i\vth_n}\partial_y A_{1,n}^{(1)}\big).
\notag
\end{align}
Inserting  
$A_{1,n}^{(2)}=-\rho_n A_{1,n}^{(1)}$,
and noting that \eqref{disprelvth} gives 
\begin{align*}
\om_\pm^\prime(\vth)
=
 \frac{- v_{1,1} v_{2,1} \sin\vth}{\om_\pm(\vth)\big(2\om_\pm^2(\vth)-c_1-c_2\big)},
\end{align*}
we obtain from the equality of the right hand sides of \eqref{relations2}
\begin{equation}\label{noninterA}
\partial_\tau A_{1,n}^{(1)}-\om_\pm^\prime(\vth_n)\partial_y A_{1,n}^{(1)}=0
\quad\text{for}\quad 
\om_n = \om_\pm (\vth_n).
\end{equation}
Analogously, in the case $\vth_n=\pm\pi$ we get from \eqref{vthpi} (for $A=A_{1,n}$)
\begin{align}\label{36}
&
\partial_\tau A^{(1)}_{1,n}=0,
\quad
A^{(2)}_{2,n}=\frac{v_{2,1}}{c_2{-}c_1} \partial_y A_{1,n}^{(1)}
\quad\text{for $\om_n^2=\om_-^2(\pm\pi)=c_1$,}
\\
&\label{36a}
\partial_\tau A^{(2)}_{1,n}=0,
\quad
A_{2,n}^{(1)}=\frac{v_{1,1}}{c_2{-}c_1}\partial_y A_{1,n}^{(2)}
\quad\text{for $\om_n^2=\om_+^2(\pm\pi)=c_2$.}
\end{align}
Thus, we conclude that if the non-resonance conditions \eqref{regularity} hold, which means  
in particular that neither wave generates a new one via self-interaction,
the dynamics of the amplitudes $A_{1,n}$ are given by uncoupled transport equations where the velocity is the group velocity of the corresponding carrier wave. Hence, setting in particular $A_{1,2}(0,\cdot)=0$ we obtain that the dynamics of $A_{1,1}$ are given, unsurprisingly, by a homogeneous transport equation. Moreover, since the $K_\iota$, $\iota\in I$, are known, as they depend only on the first order amplitudes $A_{1,n}$ (see Appendix), and since 
\eqref{relations2}, \eqref{36}$_2$, \eqref{36a}$_2$ determine the relation between the components of $A_{2,n}$, we obtain by \eqref{regularity} all $A_{2,\iota}$ except for one component of $A_{2,n}$, which can be assumed to be equivalently vanishing.  

However, it is possible that
$(\om_2,\vth_2)=(2\om_1,2\vth_1)$, 
i.e.\ $\boe_2=\boe_1^2$, 
namely for $\om_1=\om_-(\vth_1)$, $\om_2=\om_+(\vth_2)$,
which moreover implies that $\boe_1^3=\boe_1\boe_2$, $\boe_1^4=\boe_2^2$ 
do not characterize plane waves, as we have shown in Section \ref{resonances}.
In this case the formal expansion gives
\begin{align*}
&\ve\Big\{
\sum_{n=1}^2 H(\om_n,\vth_n) A_{1,n}\boe_n+\cc\Big\}
\displaybreak[0]\\&
+\ve^2\Big\{
\Big(
\begin{pmatrix}
-2\i\om_1 \partial_\tau A^{(1)}_{1,1} +v_{1,1}e^{\i\vth_1}\partial_y A_{1,1}^{(2)}
\\
-2\i\om_1 \partial_\tau A^{(2)}_{1,1} - v_{2,1}e^{-\i\vth_1} \partial_y A_{1,1}^{(1)}
\end{pmatrix}
+ H(\om_1,\vth_1) A_{2,1} +\bar K_{(1,-2)} \Big)\boe_1
\displaybreak[0]\\&\phantom{+\ve^2\Big\{\ }
+\Big(
\begin{pmatrix}
-2\i\om_2 \partial_\tau A^{(1)}_{1,2} +v_{1,1}e^{\i\vth_2}\partial_y A_{1,2}^{(2)}
\\
-2\i\om_2 \partial_\tau A^{(2)}_{1,2} - v_{2,1}e^{-\i\vth_2} \partial_y A_{1,2}^{(1)}
\end{pmatrix}
+ H(\om_2,\vth_2) A_{2,2} +K_{(1,1)} \Big)\boe_2
\displaybreak[0]\\&\phantom{+\ve^2\Big\{\ }
+\big( H(2\om_2,2\vth_2) A_{2,(2,2)} +K_{(2,2)} \big)\boe_1^4
+\big(H(\om_1{+}\om_2,\vth_1{+}\vth_2) A_{2,(1,2)} +K_{(1,2)} \big)\boe_1^3
\displaybreak[0]\\&\phantom{+\ve^2\Big\{\  }
+\frac12 H(0,0) A_{2,(1,-1)} +K_{(1,-1)} +\cc\Big\}
+\res\big(U^{A,2}_\ve\big)
=0
\end{align*}
The equations for $\ve\boe_n$ are the same as before, and hence 
\eqref{relations} and \eqref{vthpi} (with $A_{1,n}$, $\rho_n$, $\om_n$, $\vth_n$) 
are still valid. 
Then, using $\rho_n$, 
we obtain from the equations for $\ve^2\boe_n$ in the case $\vth_n\neq\pm\pi$
\begin{align*}
\rho_1 A_{2,1}^{(1)}+ A_{2,1}^{(2)}
&=\frac1{ v_{1,1}(e^{\i\vth_1}{+}1)}
\Big( 2\i\om_1 \partial_\tau A^{(1)}_{1,1} - v_{1,1}e^{\i\vth_1}\partial_y A_{1,1}^{(2)}
+\bar K_{(1,-2)}^{(1)}\Big)
\displaybreak[0]\\
&=\frac1{\om_1^2{-}c_2}
\Big(2\i\om_1 \partial_\tau A^{(2)}_{1,1} + v_{2,1}e^{-\i\vth_1} \partial_y A_{1,1}^{(1)}
+\bar K_{(1,-2)}^{(2)}\Big),
\displaybreak[0]\\
\rho_2 A_{2,2}^{(1)} +  A_{2,2}^{(2)} 
&=\frac1{v_{1,1}(e^{\i\vth_2}{+}1)}
\Big(2\i\om_2 \partial_\tau A^{(1)}_{1,2} - v_{1,1}e^{\i\vth_2}\partial_y A_{1,2}^{(2)}
+K_{(1,1)}^{(1)}\Big)
\displaybreak[0]\\
&=\frac1{ \om_2^2{-}c_2}
\Big( 2\i\om_2 \partial_\tau A^{(2)}_{1,2} + v_{2,1}e^{-\i\vth_2} \partial_y A_{1,2}^{(1)}
+K_{(1,1)}^{(2)}\Big),
\end{align*}
and inserting \eqref{relations} into the equalities on the right hand side  we get
for 
$\vth_1\neq\pm\frac\pi2,\pm\pi$,
$\vth_2=2\vth_1$, $\om_1=\om_-(\vth_1)$, $\om_2=\om_+(\vth_2)=2\om_1$
\begin{equation}\label{interA}
\begin{cases}
\partial_\tau A^{(1)}_{1,1}-\om_-^\prime(\vth_1)\partial_y A_{1,1}^{(1)}
&=\frac{d_1}{ \i \om_1 }\frac{\om_1^2{-}c_2}{(\om_1^2{-}c_1)+(\om_1^2{-}c_2)}
\bar A^{(1)}_{1,1} A^{(1)}_{1,2},
\\
\partial_\tau A^{(1)}_{1,2}-\om_+^\prime(\vth_2)\partial_y A_{1,2}^{(1)}
&=\frac{d_2}{2\i\om_2 }\frac{\om_2^2{-}c_2}{(\om_2^2{-}c_1)+(\om_2^2{-}c_2)}
\big(A_{1,1}^{(1)}\big)^2
\end{cases}
\end{equation}
with 
\begin{align*}
d_1:=&\ 
d v_{2,2} \Big(\frac{e^{-\i\vth_1}{-}1}{\bar\rho_1\rho_2} 
+\frac{e^{-\i2\vth_1}{-}1}{\rho_2} +\frac{e^{\i\vth_1}{-}1}{\bar\rho_1} \Big) 
\\&\ 
+v_{1,2} \Big(\bar\rho_1\rho_2(e^{\i\vth_1}{-}1)
+\rho_2(e^{\i2\vth_1}{-}1)+\bar\rho_1(e^{-\i\vth_1}{-}1)\Big)
+d w_{2,2} 
- w_{1,2},
\\
d_2:=&\ 
d v_{2,2} \Big(\frac{e^{-\i2\vth_1}{-}1}{\rho_1^2} +2 \frac{e^{-\i\vth_1}{-}1}{\rho_1}\Big) 
+v_{1,2} \Big(\rho_1^2(e^{\i 2\vth_1}{-}1)+2\rho_1(e^{\i\vth_1}{-}1)\Big) 
\\&
+d w_{2,2}  
-w_{1,2}, 
\\
d:=&\ \frac{v_{1,1}v_{2,1}^2(2+4\cos\vth_1+2\cos\vth_2)}{(\om_1^2{-}c_2)^2(\om_2^2{-}c_2)}.
\end{align*}

Analogously, for $\vth_1=\pm\frac\pi2$, $\vth_2=2\vth_1$, $\om_1=\om_-(\vth_1)$, 
$\om_2=\om_+(\vth_2)=\sqrt{c_2}=2\om_1$ we obtain 
\begin{align}\label{38}
\begin{cases}
\partial_\tau A^{(1)}_{1,1} 
-\om_-^\prime(\pm{\textstyle{\frac\pi2}}) 
\partial_y A_{1,1}^{(1)}
&=\frac{d_1}{\i\om_1}
\frac{\om_1^2{-}c_2}{\om_1^2-c_2 +\om_1^2-c_1}
\bar A_{1,1}^{(1)} A_{1,2}^{(2)}
\\
\partial_\tau A^{(2)}_{1,2}
&=\frac1{2\i\om_2}\Big(2 v_{2,2}\big(1+\rho_1(1{\pm}\i)\big)-w_{2,2}\rho_1^2\Big) 
\big(A^{(1)}_{1,1}\big)^2 
\end{cases}
\end{align}
with $d_1=\Big( \frac{v_{2,2}}{v_{2,1}}\mp\i\frac{w_{2,2}}{\om_1^2-c_2}\Big) 
(\om_1^2-c_1)+v_{1,2}\big(\rho_1(1{\pm}\i)+2\big)$ and
\begin{align*}
\rho_1 A_{2,1}^{(1)} + A_{2,1}^{(2)}
&=
\rho_1\Big(\frac{1\mp\i}2\partial_y A_{1,1}^{(1)}
-\frac{2\i\om_1}{\om_1^2{-}c_2}\partial_\tau A^{(1)}_{1,1}\Big) 
+2 \rho_1 \Big( \frac{v_{2,2}}{v_{2,1}}\mp\i \frac{w_{2,2}}{\om_1^2{-}c_2}\Big) 
\bar A^{(1)}_{1,1} A^{(2)}_{1,2} 
\\
A_{2,2}^{(1)}
&=\frac1{c_2{-}c_1}\Big(
v_{1,1} \partial_y A_{1,2}^{(2)}
+\Big(2 v_{1,2}\big(\rho_1^2+\rho_1(1{\mp}\i)\big)+w_{1,2}\Big) 
\big(A^{(1)}_{1,1}\big)^2
\Big)
\end{align*}
and for $\vth_1=\pm\pi$, $\vth_2=2\vth_1$, $\om_1=\om_-(\vth_1)=\sqrt{c_1}$,
$\om_2=\om_+(\vth_2)=\om_+(0)=2\om_1$ we get, 
using \eqref{vthpi} and \eqref{relations},
the equations \eqref{interA} with $\om_-^\prime(\vth_1)=\om_+^\prime(\vth_2)=0$
and $d_1=d_2=-w_{1,2}$,
and 
\begin{align*}
A_{2,1}^{(2)}
&=\frac1{c_2-c_1}\Big(v_{2,1} \partial_y A_{1,1}^{(1)}+4 v_{2,2}  (1+\rho_2) \bar A_{1,1}^{(1)}  A_{1,2}^{(1)}\Big), 
\displaybreak[0]\\
\rho_2 A_{2,2}^{(1)}+A_{2,2}^{(2)}
&=
\rho_2\Big(\frac12\partial_y A_{1,2}^{(1)}
-\frac{2\i\om_2}{\om_2^2{-}c_2}\partial_\tau A^{(1)}_{1,2}\Big). 
\end{align*}

Hence, in the case $\boe_2=\boe_1^2$ we obtain two coupled equations for $A_{1,n}^{(1[2, \text{resp.}])}$, 
the solutions of which (cf.\ about their well-posedness Lemma \ref{lemmares}) determine again 
$U_\ve^{A,2}$ up to one component of $A_{2,n}$.
We would like to stress that in order to obtain non-trivial dynamics for $A_{1,1}$ we have to consider also the dynamics of the generated wave $A_{1,2}$ while if only interested in $A_{1,2}$ we could ignore the generating wave $A_{1,1}$, 
see e.g.\ \eqref{interA}. Note in this context, that even for initial data $A_{1,2}(0,\cdot)=0$ an amplitude $A_{1,2}\not\equiv0$ emerges, which motivates the notion of \emph{generation}
of waves. 

\section{Justification}\label{justification}
The equations obtained by the formal derivation constitute only necessary conditions on the amplitudes $A_{1,n}$ of the \emph{ansatz} \eqref{appr}. The purpose of the justification is to show that indeed solutions $u$ of such a form exist.  
\begin{theorem}\label{justiftheorem}
Let $V_i,W_i\in C^{4}(\R)$, $i=1,2$,  in \eqref{TEVW} satisfy
\begin{equation}\label{stabilityassumption}
v_{1,1}=\frac{v_1}M,\ v_{2,1}=\frac{v_1}m,\ 
w_{i,1}>0,\ 
4 v_1+\min\{M w_{1,1}, m w_{2,1}\}>0,\
M,m>0,
\end{equation}
let  
$\om_2>\om_1>0$ and $\vth_n\in(-\pi,\pi]$, $n=1,2$, 
satisfy $\det H(\om_n,\vth_n)=0$ and 
\begin{align*}
\text{either}\quad & 
\det H(2\om_n,2\vth_n)\neq 0,\quad 
\det H(\om_1\pm\om_2,\vth_1\pm\vth_2)\neq 0,
\\
\text{or}\quad &
(\om_2,\vth_2)=(2\om_1,2\vth_1\ \mathrm{mod}\ 2\pi),\quad
\det H(k\om_1, k\vth_1)\neq 0,\ k=3,4, 
\end{align*}
with the dispersion matrix $H$ in \eqref{H},
and let $A_{1,n}^{(1[2])}:[0,\tau_0]\times\R\to\C$, $\tau_0>0$, 
be, respectively, the unique solutions of either \eqref{noninterA}(or \eqref{36} or \eqref{36a}) 
or \eqref{interA} (or \eqref{38})
with  $A_{1,n}^{(1[2])}(0,\cdot)\in H^4(\R;\C)$.

Then, for the corresponding approximation $U^{A,1}_\ve$ and every $c>0$, $\beta\in\big(1,3/2]$ 
there exist $\ve_0,C>0$ 
such that for all $\ve\in(0,\ve_0)$ and all $t\in[0,\tau_0/\ve]$ 
any solution $u$ of \eqref{dichainS} satisfies
\begin{align*}
\bigg\|\begin{pmatrix} u - U^{A,1}_\ve\\ \dot u - \dot U^{A,1}_\ve \end{pmatrix}(0)
\bigg\|_{(\ell^2)^4}
\le c\ve^\beta
\quad\Rightarrow\quad 
\bigg\|\begin{pmatrix} u - U^{A,1}_\ve \\ \dot u - \dot U^{A,1}_\ve \end{pmatrix}(t)
\bigg\|_{(\ell^2)^4}
\le C\ve^\beta.
\end{align*}
\end{theorem}
\begin{proof}
The idea of the proof is classical, see e.g.\ \cite{KSM9?}.
We write the microscopic model \eqref{dichainS}
as a first order system in  
$Y:=\big(\ell^2\big)^4$
with $(\ell^2)^4=(\ell^2)^2\times(\ell^2)^2=(\ell^2\times\ell^2)\times(\ell^2\times\ell^2)$
and $\ell^2=\ell^2(\Z)$,
\begin{equation}\label{dichainStilde}
\dot{\wt u}=\wt \calL\wt u+\wt \calM(\wt u)
\quad\text{with}\quad
\wt u:=\begin{pmatrix} u \\ \dot u \end{pmatrix},\quad
\wt \calL:=\begin{pmatrix}0 & \calI \\ \calL & 0\end{pmatrix},\quad
\wt \calM(\wt u):=\begin{pmatrix} 0 \\ \calM(u) \end{pmatrix},
\end{equation}
where $\calI:\big(\ell^2\big)^2\to\big(\ell^2\big)^2$ is the identity.
Then, the flow of the linearized system $\dot{\wt u}=\wt\calL\wt u$ 
preserves the energy norm on $Y$, 
\begin{align*}
\|\check u\|_Y^2:=&\|u\|_E^2+\|\tu\|_M^2
\\
=&\sum_{j\in\Z} \Big(v_1\big(|u_{j+1,2}-u_{j,1}|^2+|u_{j,1}-u_{j,2}|^2\big)
+ M w_{1,1} |u_{j,1}|^2+ m w_{2,1}|u_{j,2}|^2\Big)
\\&
+\sum_{j\in\Z}\big( M |\tu_{j,1}|^2 + m |\tu_{j,2}|^2\big) 
\quad\text{for}\ \check u=\begin{pmatrix} u \\ \tu \end{pmatrix}, 
\end{align*}
i.e.\ its associated semi-group $e^{t\wt\calL}$   
satisfies $\| e^{t\wt \calL}\|_{Y\to Y}=1$,
and from \eqref{stabilityassumption} 
it follows by Fourier transformation that the norms 
$\|{\cdot}\|_E$, $\|{\cdot}\|_M$ and $\|{\cdot}\|_{(\ell^2)^2}$,
and hence also $\|{\cdot}\|_Y$ and $\|{\cdot}\|_{(\ell^2)^4}$, 
are equivalent: 
$\hat\kappa_1\|u\|_{(\ell^2)^2}\le\|u\|_M \le\hat\kappa_2\|u\|_{(\ell^2)^2}$, 
$\kappa_1\|u\|_{(\ell^2)^2}\le\|u\|_E\le\kappa_2\|u\|_{(\ell^2)^2}$,  
and  $\check\kappa_1\|\check u\|_{(\ell^2)^4}\le\|\check u\|_Y
\le\check\kappa_2\|\check u\|_{(\ell^2)^4}$,  
with $\hat\kappa_i, \kappa_i, \check\kappa_i>0$
and $\|u\|_{(\ell^2)^2}^2=\|u_1\|_{\ell^2}^2+\|u_2\|_{\ell^2}^2$  for $u=(u_1,u_2)^T$,
$\|\check u\|_{(\ell^2)^4}^2=\|u\|_{(\ell^2)^2}^2+\|\tu\|_{(\ell^2)^2}^2$.

We consider the error 
$\ve^\beta \wt R_\ve :=\wt u-\wt U_{\ve}^{A,2}
=\big( u - U_{\ve}^{A,2},\dot u - \dot U_{\ve}^{A,2}\big)^T$
between an original solution $u$ of \eqref{dichainS} and 
the improved approximation $U_{\ve}^{A,2}$ given by \eqref{imprappr} 
with the $A_{1,n}$, $A_{2,\iota}$ determined by the formal derivation.
Since for this $U_\ve^{A,2}$ we have 
$\calL U_\ve^{A,2}+\calM\big(U_\ve^{A,2}\big)-\ddot U_\ve^{A,2}=\res\big(U_\ve^{A,2}\big)$, 
inserting $\wt R_\ve$ into \eqref{dichainStilde}
we obtain the differential equation
\begin{align*}
\dot{\wt R}_\ve
&=
\wt \calL\wt R_\ve
+\ve^{-\beta}\begin{pmatrix} 0 \\ 
\calM\big(U_\ve^{A,2}+\ve^\beta R_\ve\big)-\calM\big(U_\ve^{A,2}\big)
+\res\big(U_\ve^{A,2}\big) \end{pmatrix}.
\end{align*}
Taking the energy norm of its integral formulation, assuming   
$\big\|\wt R_\ve(0)\big\|_Y\le d$, and applying Lemma \ref{lemmares} c), 
we get
\begin{multline}\label{varconst}
\big\|\wt R_\ve(t)\big\|_Y
\le  d + \ve^{3/2-\beta} \tau_0 c_r
+\ve^{-\beta}\int_0^t
\big\|\calM\big(U_\ve^{A,2}{+}\ve^\beta R_\ve\big)
-\calM\big(U_\ve^{A,2}\big)\big\|_M \,d s
\end{multline}
for $\ve\in(0,\ve_0]$, $t\in[0,\tau_0/\ve]$.
From \eqref{dichainS} and \eqref{TEVW} we get by the mean value theorem
\begin{align}\label{estNgeneral}
\|\calM(u)-\calM(\tu)\|_M
\le  c_\calM(\|u\|_{\ell^\infty}+\|\tu\|_{\ell^\infty}) \|u-\tu\|_M
\quad\text{for $\|u\|_{\ell^\infty}, \|\tu\|_{\ell^\infty}\le c_0$}
\end{align}
with 
$\|u\|_{\ell^\infty}=\max\{\|u_1\|_{\ell^\infty},\|u_2\|_{\ell^\infty}\}$ 
for $u=(u_1,u_2)^T\in(\ell^2)^2$, 
and $c_\calM$ depending only on $V_i,W_i$ and $c_0>0$.

We set 
$D:=\big(d+\ve_0^{3/2-\beta} \tau_0 c_r \big) e^{\tau_0 3c_\calM c_A\hat\kappa_2/\kappa_1}$
with $c_A$ from Lemma \ref{lemmares} a)
and $\ve_0>0$ such that $\ve_0^\beta D/\kappa_1\le \ve_0 c_A\le c_0/2$.
Since $\big\|\wt R_\ve(0)\big\|_Y\le d<D$ and $\big\|\wt R_\ve(t)\big\|_Y$ is continuous,
there exists for every $\ve\in(0,\ve_0]$ a $t_D^\ve>0$, such that
$\big\|\wt R_\ve(t)\big\|_Y\le D$ for $t\in[0,t_D^\ve]$.
Then, 
for $\ve\in(0,\ve_0]$ and $t\in[0,\min\{\tau_0/\ve,t_D^\ve\}]$
\eqref{estNgeneral} gives
\begin{equation*}
\big\|\calM\big(U_\ve^{A,2}{+}\ve^\beta R_\ve\big)-\calM\big(U_\ve^{A,2}\big)\big\|_M
\le \ve^{\beta+1} (3 c_\calM c_A\hat\kappa_2/\kappa_1) \big\|\wt R_\ve\big\|_Y.
\end{equation*}
Inserting this estimate into \eqref{varconst} and applying Gronwall's Lemma, we get
\begin{align*}
\big\|\wt R_\ve(t)\big\|_Y
\le \Big(d+\ve_0^{3/2-\beta}\, \tau_0\, c_r\Big) e^{\ve t 3 c_\calM c_A\hat\kappa_2/\kappa_1}
\le D
\quad\text{for $\ve\in(0,\ve_0]$, 
$t\in[0,\tau_0/\ve]$.}
\end{align*}

Finally, 
with 
$d:=\check\kappa_2 c + \ve_0^{3/2-\beta}c_I$ and $C:=(D+\ve_0^{3/2-\beta}c_I)/\check\kappa_1$
we obtain 
from Lemma \ref{lemmares} b) and the equivalence of $\|\cdot\|_{(\ell^2)^4}$ and $\|\cdot\|_Y$
the assertion of the theorem. \end{proof}

\begin{lemma}\label{lemmares}
For $U_\ve^{A,2}$ given by \eqref{imprappr} 
with $A_{1,n}$, $A_{2,\iota}$ as determined in Section \ref{formalderivation}
and initial data $A_{1,n}^{(i)}(0,\cdot)\in H^4(\R;\C)$,
there exist $\tau_0, \ve_0, c_A, c_I, c_r>0$ 
such that for all $\ve\in[0,\ve_0]$, $\ve t\in[0,\tau_0]$
\begin{align*}
\text{a)}\ \big\|U_\ve^{A,2}\big\|_{\ell^\infty} \le\ve c_A,
\quad
\text{b)}\ \big\|\wt U_\ve^{A,2}{-}\wt U_\ve^{A,1}\big\|_Y \le\ve^{3/2}c_I,
\quad
\text{c)}\ \big\|\res\big(U_\ve^{A,2}\big)\big\|_M \le \ve^{5/2} c_r. 
\end{align*}
\end{lemma}
\begin{proof}
Inserting into \eqref{appr}, \eqref{imprappr} and $\res\big(U_\ve^{A,2}\big)$ (see Appendix) 
$A_{1,n}^{(2)}=-\rho_n A_{1,n}^{(1)}$
(with $\rho_n=0$ for 
$\vth_n=\pm\pi$, $\om_n=\sqrt{c_1}$),
$A_{2,n}^{(1)}\equiv0$ 
[or $A_{1,n}^{(1)}\equiv A_{2,n}^{(2)}\equiv  0$ for 
$\vth_n=\pm\pi$, $\om_n=\sqrt{c_2}$],
and the $A_{2,n}^{(2)}$ [or  $A_{2,n}^{(1)}$, respectively], $A_{2,\iota}$, $\iota\in I$, specified in Section \ref{formalderivation},
recalling $|\wt V_i^\prime(x)|,|\wt W_i^\prime(x)|=O(|x|^3)_{x\to0}$
and the equivalence of 
$\|{\cdot}\|_E$, $\|{\cdot}\|_M$, $\|{\cdot}\|_{(\ell^2)^2}$, 
and 
using the corollary 
of  Sobolev's embedding theorem (cf., e.g., \cite[Lemma 3.1]{Gia10})
\begin{equation*}
\big\|\varphi\big(\ve(\cdot+\xi)\big)\big\|_{\ell^2} 
\le c \ve^{-1/2} \|\varphi\|_{H^1(\R;\C)},
\quad\xi:\Z\to[-1,1],\ \ve\ \in(0,\ve_0],
\end{equation*}
and 
$A,B\in C([0,\tau_0]; H^1(\R;\C ))\ \Rightarrow\ AB\in C([0,\tau_0]; H^1(\R;\C ))$, 
we obtain 
that the above estimates are satisfied, provided
\begin{align*}
\partial_\tau^p \partial_y^q A_{1,n}^{(1[2,\text{resp.}])} 
\in C([0,\tau_0]; H^1(\R;\C ))
\quad\text{for $|(p,q)|\le 2$, $(p,q)=(3,0),(2,1)$.}
\end{align*}
Since the macroscopic equations for $A_{1,n}^{(1[2])}$, $n=1,2$, are 
semilinear autonomous transport systems with smooth nonlinearities, 
standard results of semigroup theory (cf., e.g., \cite[Th.\ 6.1.7]{Paz83}) yield
that for initial data $A_{1,n}^{(1[2])}(0,\cdot)\in H^m(\R;\C)$, $m\ge1$, 
there exist unique classical solutions 
with $\partial_\tau^p \partial_y^q A_{1,n}^{(1[2])}\in C([0,\tau_0]; H^1(\R;\C))$
for $|(p,q)|\le m-1$
up to some $\tau_0\in(0,\infty]$. 
Hence, for $m=4$ 
we obtain the statement of the lemma.
\end{proof}

\section{Appendix}
For completeness we present here the 
$K_\iota=\big(K_\iota^{(1)},K_\iota^{(2)}\big)^T$, 
$\iota\in I$, and 
$\res\big(U_\ve^{A,2}\big)=\big(\res_1\big(U_\ve^{A,2}\big), \res_2\big(U_\ve^{A,2}\big)\big)^T$
derived in Section \ref{formalderivation}, 
with $i=1,2$ (where $i+1=1$ for $i=2$) and with the upper sign corresponding to $i=1$.
\begin{align*}
K_{(n,n)}^{(i)}:=&
\pm 
v_{i,2}\Big(\big(A_{1,n}^{(i+1)}\big)^2(e^{\pm\i 2\vth_n}{-}1)
-2 A_{1,n}^{(1)}A_{1,n}^{(2)}(e^{\pm\i\vth_n}{-}1)\Big)
-w_{i,2} \big(A^{(i)}_{1,n}\big)^2,
\displaybreak[0]\\
K_{(1,2)}^{(i)}:=&
\pm 
2 v_{i,2}
\Big(
A_{1,1}^{(i+1)}A_{1,2}^{(i+1)}(e^{\pm\i(\vth_1+\vth_2)}{-}1)
-A_{1,1}^{(i)}A_{1,2}^{(i+1)}(e^{\pm\i\vth_2}{-}1)
\displaybreak[0]\\&
-A_{1,1}^{(i+1)}A_{1,2}^{(i)}(e^{\pm\i\vth_1}{-}1)
\Big)
- 2 w_{i,2} A^{(i)}_{1,1} A^{(i)}_{1,2},
\displaybreak[0]\\
K_{(1,-2)}^{(i)}:=&
\pm 
 2 v_{i,2}
\Big(
 A_{1,1}^{(i+1)}\bar A_{1,2}^{(i+1)}(e^{\pm\i(\vth_1-\vth_2)}{-}1)
-A_{1,1}^{(i)}\bar A_{1,2}^{(i+1)}(e^{\mp\i\vth_2}{-}1)
\displaybreak[0]\\&
-A_{1,1}^{(i+1)}\bar A_{1,2}^{(i)}(e^{\pm\i\vth_1}{-}1)
\Big)
-2 w_{i,2} A^{(i)}_{1,1}\bar A^{(i)}_{1,2},
\displaybreak[0]\\
K_{(1,-1)}^{(i)}:=&
\sum_{n=1}^2 \Big(
\mp
2 v_{i,2} \bar A_{1,n}^{(i)} A_{1,n}^{(i+1)}(e^{\pm\i\vth_n}{-}1) 
-w_{i,2} \big|A^{(i)}_{1,n}\big|^2
\Big);
\end{align*}
\begin{align*} 
&
\res_i\big(U_\ve^{A,2}\big):=
\ve^3\Big( -T_i \pm v_{i,1}F_i \pm 2 v_{i,2}\big(D_i E_i{-}(A_1{-}A_2)(B_1{-}B_2)\big)
-2 w_{i,2} A_i B_i \Big)
\displaybreak[0]\\&\quad
+\ve^4\Big( -S_i \pm v_{i,2}\big(E_i^2{+}2 D_i F_i{-}(B_1{-}B_2)^2\big) -w_{i,2}B_i^2\Big)
\pm \ve^5 v_{i,2} E_i F_i \pm \ve^6 v_{i,2} F_i^2
\displaybreak[0]\\&\quad
\pm\widetilde V_i^\prime \big(\ve D_i{+}\ve^2 E_i{+}\ve^3 F_i\big)
\mp \widetilde V_i^\prime \big(
\ve (A_1{-}A_2){+}\ve^2 (B_1{-}B_2)
\big)
-\widetilde W_i^\prime 
(\ve A_i{+}\ve^2 B_i),
\end{align*}
\begin{align*}
T_i:= &
\sum_{n=1}^2 
\Big(\big(\partial_\tau^2 A^{(i)}_{1,n}{+}\partial_\tau A^{(i)}_{2,n} 2\i\om_n\big) \boe_n 
+\partial_\tau A^{(i)}_{2,(n,n)} 4\i\om_n \boe_n^2\Big) 
\displaybreak[0]\\&
+\partial_\tau A^{(i)}_{2,(1,2)}2\i(\om_1{+}\om_2 )\boe_1\boe_2 
+\partial_\tau A^{(i)}_{2,(1,-2)}2\i(\om_1{-}\om_2)\boe_1\boe_{-2}
+\cc,
\displaybreak[0]\\
S_i:= &
\sum_{n=1}^2 \Big(\partial_\tau^2 A^{(i)}_{2,n}\boe_n 
+\partial_\tau^2 A^{(i)}_{2,(n,n)}\boe_n^2 \Big)
+\partial_\tau^2 A^{(i)}_{2,(1,2)}\boe_1\boe_2 
+\partial_\tau^2 A^{(i)}_{2,(1,-2)}\boe_1\boe_{-2} 
\displaybreak[0]\\&
+\frac12\partial_\tau^2 A^{(i)}_{2,(1,-1)}
+\cc,
\displaybreak[0]\\
A_i:= &
\sum_{n=1}^2 A_{1,n}^{(i)} \boe_n+\cc,
\qquad 
D_i:= \pm\sum_{n=1}^2\big(A_{1,n}^{(i+1)}e^{\pm\i\vth_n}{-}A_{1,n}^{(i)}\big)\boe_n+\cc,
\displaybreak[0]\\
B_i:= & \sum_{n=1}^2 \Big(A_{2,n}^{(i)}\boe_n+ A_{2,(n,n)}^{(i)}\boe_n^2\Big)
+A_{2,(1,2)}^{(i)}\boe_1\boe_2+A_{2,(1,-2)}^{(i)}\boe_1\boe_{-2}+\frac12 A_{2,(1,-1)}^{(i)}
\displaybreak[0]\\&
+\cc,
\displaybreak[0]\\
E_i:= &
\sum_{n=1}^2 \Big(
\big(\partial_y A_{1,n}^{(i+1)}e^{\pm\i\vth_n}
{\pm} A_{2,n}^{(i+1)}e^{\pm\i\vth_n}{\mp} A^{(i)}_{2,n}\big)\boe_n
\pm \big(A_{2,(n,n)}^{(i+1)}e^{\pm\i 2\vth_n}{-}A^{(i)}_{2,(n,n)}\big)\boe_n^2\Big)
\displaybreak[0]\\&
\pm \big(A_{2,(1,2)}^{(i+1)}e^{\pm\i(\vth_1+\vth_2)}{-}A^{(i)}_{2,(1,2)}\big)\boe_1\boe_2
\pm \big(A_{2,(1,-2)}^{(i+1)}e^{\pm \i(\vth_1-\vth_2)}{-}A^{(i)}_{2,(1,-2)}\big)\boe_1\boe_{-2}
\displaybreak[0]\\&
+\frac12 \big(A^{(2)}_{2,(1,-1)}{-}A^{(1)}_{2,(1,-1)}\big)
+\cc,
\displaybreak[0]\\
F_i:= &
\sum_{n=1}^2 \Big(
\big(\pm\frac12\partial_y^2 A_{1,n,\xi_1\pm}^{(i+1)}
{+}\partial_y A_{2,n,\xi_2\pm}^{(i+1)}\big)e^{\pm\i\vth_n} \boe_n
+\partial_y A_{2,(n,n),\xi_2\pm}^{(i+1)}e^{\pm\i 2\vth_n}\boe_n^2\Big)
\displaybreak[0]\\&
+\partial_y A_{2,(1,2),\xi_2\pm}^{(i+1)}e^{\pm\i (\vth_1+\vth_2 )}\boe_1\boe_2
+\partial_y A_{2,(1,-2),\xi_2\pm}^{(i+1)}e^{\pm\i (\vth_1-\vth_2 )}\boe_1\boe_{-2}
\displaybreak[0]\\&
+\frac12\partial_y A_{2,(1,-1),\xi_2\pm}^{(i+1)}
+\cc
\end{align*}

\end{document}